\documentclass[sigconf,letterpaper,9pt]{acmart}

\usepackage{booktabs} 
\usepackage{amsmath,amssymb,amsthm}
\usepackage{color}
\usepackage{bbm}
\usepackage{todonotes}
\usepackage{booktabs}
\usepackage{breakcites}
\usepackage{bm}
\usepackage{natbib}
\usepackage{subcaption}
\usepackage{forest}
\usepackage{qtree}
\usepackage{subcaption}
\usepackage{tabularx}

\theoremstyle{plain}

\theoremstyle{definition}

\theoremstyle{remark}

\usepackage{subcaption}
\usepackage{amsmath}
\usepackage{bbm}
\usepackage{dsfont}

\newtheorem{thm}{Theorem}

\setcopyright{none}

\renewcommand\footnotetextcopyrightpermission[1]{} 
\makeatletter
\def\blfootnote{\gdef\@thefnmark{}\@footnotetext}
\makeatother


\begin{document}
\title{Bias in Bios: A Case Study of Semantic Representation Bias in a High-Stakes Setting}

\author[De-Arteaga et al.]{Maria De-Arteaga$^1$, Alexey Romanov$^2$, Hanna Wallach$^3$, Jennifer Chayes$^3$, Christian Borgs$^3$, Alexandra Chouldechova$^1$, Sahin Geyik$^4$, Krishnaram Kenthapadi$^4$, Adam Tauman Kalai$^3$\\ \vspace{0.05in}
\textit{$^1$Carnegie Mellon University, $^2$University of Massachusetts Lowell, $^3$Microsoft Research, $^4$LinkedIn}}

\begin{abstract}
We present a large-scale study of gender bias in occupation classification, a task where the use of machine learning may lead to negative outcomes on peoples' lives. We analyze the potential allocation harms that can result from semantic representation bias. To do so, we study the impact on occupation classification of including explicit gender indicators---such as first names and pronouns---in different semantic representations of online biographies. Additionally, we quantify the bias that remains when these indicators are ``scrubbed,'' and describe proxy behavior that occurs in the absence of explicit gender indicators. As we demonstrate, differences in true positive rates between genders are correlated with existing gender imbalances in occupations, which may compound these imbalances.\looseness=-1
\end{abstract}


\maketitle

\section{Introduction}

\blfootnote{This paper has been accepted for publication in the \textbf{ACM Conference on Fairness, Accountability, and 
Transparency (ACM FAT*), 2019.}}

The presence of automated decision-making systems in our daily lives is growing. As a result these systems play an increasingly active role in shaping our future. Far from being passive players that consume information, automated decision-making systems are participating actors: their predictions today affect the world we live in tomorrow. In particular, they determine many aspects of how we experience the world, from the news we read and the products we shop for to the job postings we see. The increased prevalence of machine learning has therefore been accompanied by a growing concern regarding the circumstances and mechanisms by which such systems may reproduce and augment the various forms of discrimination and injustices that are present in today's society.

One domain in which the use of machine learning is increasingly popular---and in which unfair practices can lead to particularly negative consequences---is that of online recruiting and automated hiring. Maintaining an online professional presence has become increasingly important for people's careers, and this information is often used as input to automated decision-making systems that advertise open positions and recruit candidates for jobs and other professional opportunities. In order to perform these tasks, a system must be able to accurately assess people's current occupations, skills, interests, and ``potential.'' However, even the simplest of these tasks---determining someone's current occupation---can be non-trivial. Although this information may be provided in a structured form on some professional networking platforms, this is not always the case. As a result, recruiters often browse candidates' websites in an attempt to manually determine their current occupations. Machine learning promises to reduce this burden; however, as we will explain in this paper, occupation classification is susceptible to gender bias, stemming from existing gender imbalances in occupations.\looseness=-1

To study gender bias in occupation classification, we created a new dataset of hundreds of thousands of online biographies, written in English, from the Common Crawl corpus. Because biographies are typically written in the third person by their subjects (or people familiar with their subjects) and because pronouns are gendered in English, we were able to extract (likely) self-identified binary gender from the biographies. We note, though, that this binary model is a simplification that fails to capture important aspects of gender and erases people who do not fit within its assumptions.

Using this dataset, we predicted people's occupations by performing multi-class classification using three different semantic representations: bag-of-words, word embeddings, and deep recurrent neural networks. For each representation, we considered two scenarios: (1) where explicit gender indicators are available to the classifier, (2) where explicit gender indicators are ``scrubbed'' to promote fairness or to comply with regulations or laws. We define explicit gender indicators to be information, such as first names and gendered pronouns, that make it possible to determine gender. We note that the practice of ``scrubbing'' explicit gender indicators and other sensitive attributes is not unique to machine learning, and is often used as a way to mitigate the effects of implicit and explicit bias on decisions made by humans. For example, gender diversity in orchestras was significantly improved by the introduction of ``blind'' auditions, where candidates play behind a curtain~\citep{goldin2000orchestrating}.\looseness=-1

To quantify gender bias, we compute the true positive rate (TPR) gender gap---i.e., the difference in TPRs between genders---for each occupation. The TPR for a given gender and occupation is defined as the proportion of people with that gender and occupation that are correctly predicted as having that occupation. We also compute the correlation between these TPR gender gaps and existing gender imbalances in occupations, and show how this may compound these imbalances; we connect this finding with an existing notion of indirect discrimination in political philosophy. We show that ``scrubbing'' explicit gender indicators reduces the TPR gender gaps, while maintaining overall classifier accuracy. However, we also show that significant TPR gender gaps remain in the absence of explicit gender indicators, and that these gaps are correlated with existing gender imbalances. For orchestra auditions, the sounds made by candidates' shoes mean that a curtain is not sufficient to make an audition ``blind.'' It is therefore common practice to additionally roll out a carpet or to ask candidates to remove their shoes~\citep{goldin2000orchestrating}. By analogy, ``scrubbing'' explicit gender indicators is like introducing a curtain---the sounds made by the candidates' shoes remain.\looseness=-1

Our paper has two main takeaways: First, ``scrubbing'' explicit gender indicators is not sufficient to remove gender bias from an occupation classifier. Second, even in the absence of such indicators, TPR gender gaps are correlated with existing gender imbalances in occupations; occupation classifiers may therefore compound existing gender imbalances. Although we focus on gender bias, we note that other biases, such as those involving race or socioeconomic status, may also be present in occupation classification or in other tasks related to online recruiting and automated hiring. We  structure our analysis so as to inform discussions about these biases as well.\looseness=-1

In the next section, we provide a brief overview of related work. We then describe our data collection process in Section~\ref{sec:data} and outline our methodology in Section~\ref{sec:methods}, before presenting our analysis and results in Section~\ref{sec:results}. We conclude with a discussion in Section \ref{sec:discussion}. 

\section{Related Work}
\label{sec:background}

Recent work has studied the ways in which stereotypes and other human biases may be reflected in semantic representations such as word embeddings~\citep{bolukbasi2016man,caliskan2017semantics,garg2018word}. Natural language processing researchers have also studied gender bias in coreference resolution~\cite{zhao2018gender,rudinger2018gender}, showing that systems perform better when linking a gender pronoun to an occupation in which that gender is overrepresented than to an occupation in which it is underrepresented. Gender bias has also been studied in YouTube's autocaptioning~\citep{tatman2017gender}, where researchers found a higher word error rate for female speakers. In the context of language identification, researchers have also investigated racial bias, showing that African-American English is often misclassified as non-English~\citep{blodgett2017racial}. Finally, machine learning methods for identifying toxic comments exhibit disproportionately high false positive rates for words like \textit{gay} and \textit{homosexual}~\citep{dixon2017measuring}. 

In the context of structured data, there have been extensive discussions about proxy behavior that may occur when sensitive attributes are not explicitly available but can be determined from other attributes~\citep{pope2011implementing,barocas2016big,zemel2013learning}. Related discussions have focused on the phenomenon of differential subgroup validity~\citep{ayres2002outcome}, where the choice of attributes may disadvantage groups for whom the chosen attributes are not equally predictive of the target label~\citep{calders2013unbiased}. \citet{barocas2016big} discussed these issues in the context of automated hiring; \citet{kim2016data} explained how data-driven decisions that systematically bias people's access to opportunities relate to existing antidiscrimination legislation, identifying voids that may need to be filled to account for potential risks stemming from automated decision-making systems. Researchers have also discussed making available sensitive attributes as a means to improve fairness~\citep{dwork2012fairness}, as well as various ways to use these attributes~\citep{dwork2018decoupled,pope2011implementing}. Finally, although our paper does not directly consider ranking scenarios, fairness in ranking is particularly relevant to discussions about gender bias in online recruiting and automated hiring~\citep{zehlike2017fa,celis2017ranking,yang2017measuring,biega2018equity,GK18}.

We quantify gender bias by computing the TPR gender gap---i.e., the difference in TPRs between genders---for each occupation. This notion of bias is closely related to the equality of opportunity fairness metric of~\citet{hardt2016equality}. We choose to focus on TPR gender gaps because they enable us to study the ways in which gender imbalances may be compounded; in turn, we relate this to compounding injustices~\citep{hellman2017indirect}---an existing notion of indirect discrimination in political philosophy that holds that it is a general moral duty to refrain from taking actions that would harm people when those actions are informed by, and would compound, prior injustices suffered by those people. We show that the TPR gender gaps are correlated with existing gender imbalances in occupations. As a result, occupation classifiers compound injustices when existing gender imbalances are attributable to historical discrimination.

Our paper is also closely related to research on gender bias in hiring~\citep{sarsons2015gender,sarsons2017interpreting,ginther2004women,bertrand2017field}. In particular, \citet{bertrand2004emily} conducted an experiment in which they responded to help-wanted ads using fictitious resumes, varying names so as to signal gender and race, while keeping everything else the same. They were therefore able to measure the effect of (inferred) gender and race on the likelihood of being called for an interview. Similarly, we study the effect of explicit gender indicators on occupation classification.

Computational linguistics researchers have explored the use of lexical and syntactic features to infer authors' genders~\citep{cheng2011author,koppel2002automatically}. Given that our dataset consists of online biographies, our paper is also related to research on differences between the ways that men and women represent themselves. In the context of online professional presences, ~\citet{altenburger2017there} analyzed self-promotion in LinkedIn, finding that women are more modest than men in expressing accomplishments and are less likely to use free-form fields. Researchers have also studied differences in volubility between men and women~\citep{brescoll2011takes}, showing that women's fear of being highly voluble is justified by the fact that both men and women negatively evaluate highly voluble women. Moving beyond self-representation, \citet{niven2001women} analyzed congressional websites and found that differences between the ways that the media portray men and women in Congress cannot be explained by differences between the ways that they portray themselves. Meanwhile, \citet{smith2018power} analyzed attributes used to describe men and women in performance evaluations, showing that negative attributes are more often used to describe women than men. This research on representation by others relates to our paper because we cannot be sure that the online biographies in our dataset were actually written by their subjects.\looseness=-1

\section{Data Collection Process}
\label{sec:data}

To study gender bias in occupation classification, we created a new dataset using the Common Crawl. Specifically, we identified online biographies, written in English, by filtering for lines that began with a name-like pattern (i.e., a sequence of two capitalized words) followed by the string ``is a(n) (xxx) \textit{title},'' where \textit{title} is an occupation from the BLS Standard Occupation Classification system.\footnote{\url{https://www.bls.gov/soc/}} We identified the twenty-eight most frequent occupations based on their appearance in a small subset of the Common Crawl. In a few cases, we merged occupations. For example, we created the occupation \textit{professor} by merging occupations that consist of \textit{professor} and a modifier, such as \textit{economics professor}. Having identified the most frequent occupations, we processed WET\footnote{WET is a special file format containing cleaned text extracted from webpages.} files from sixteen distinct crawls from 2014 to 2018, extracting online biographies corresponding to those occupations only.
Finally, we performed de-duplication by treating biographies as duplicates if they had the same first name, last name, and occupation, and either no middle name was present or one middle name was a prefix of the other. 
The resulting dataset consists of 397,340 biographies spanning twenty-eight different occupations. Of these occupations, \textit{professor} is the most frequent, with 118,400 biographies, while \textit{rapper} is the least frequent, with 1,406 biographies (see~\autoref{fig:titles_distribution}). The longest biography is 194 tokens, while the shortest is eighteen; the median biography length is seventy-two tokens. We note that the demographics of online biographies' subjects differ from those of the overall workforce, and that our dataset does not contain all biographies on the Internet; however, neither of these factors is likely to undermine our findings.\looseness=-1

\begin{figure}[ht]
\centering
\includegraphics[width=1\linewidth, 
clip]{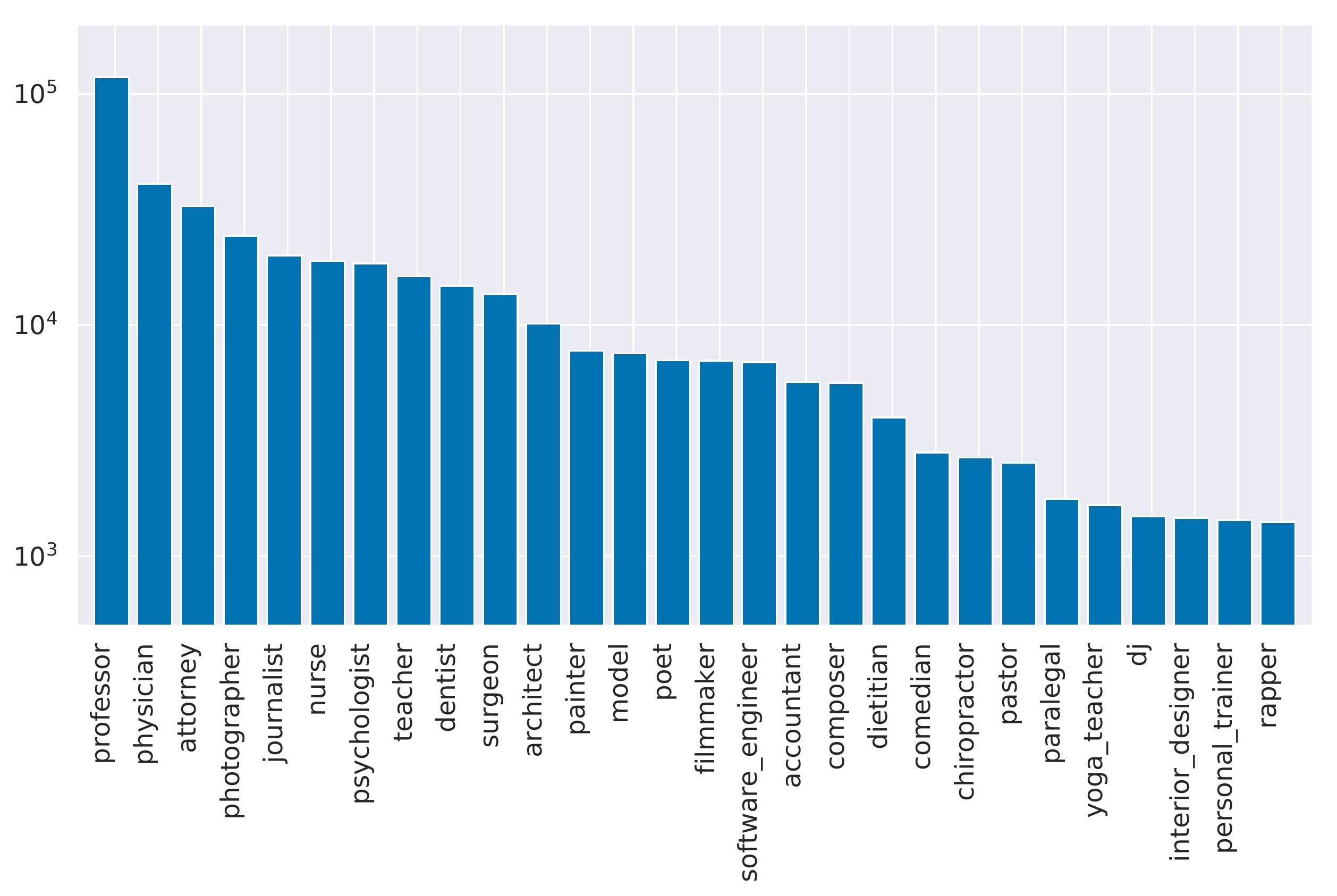} 
\caption{Distribution of the number of biographies for the twenty-eight different occupations, shown on a log scale.}
\label{fig:titles_distribution} 
\end{figure}

Because some occupations have a high gender imbalance, our validation and testing splits must be large enough that every gender and occupation are sufficiently represented. We therefore used stratified-by-occupation splits, with 65\% of the biographies (258,370) designated for training, 10\% (39,635 biographies) designated for validation, and 25\% (99,335 biographies) designated for testing.


A complete implementation that reproduces the dataset can be found in the source code available at \href{http://aka.ms/biasbios}{http://aka.ms/biasbios}.


\section{Methodology}
\label{sec:methods}


We used our dataset to predict people's occupations, taken from the first sentence of their biographies as described in the previous section, given the remainder of their biographies. For example, consider the hypothetical biography \textit{Nancy Lee is a registered nurse. She graduated from Lehigh University, with honours in 1998. Nancy has years of experience in weight loss surgery, patient support, education, and diabetes.} The goal is to predict \textit{nurse} from \textit{She graduated from Lehigh University, with honours in 1998. Nancy has years of experience in weight loss surgery, patient support, education, and diabetes.}\looseness=-1

We used three different semantic representations of varying complexity: bag-of-words (BOW), word embeddings (WE), and deep recurrent neural networks (DNN). When using the BOW and WE representations, we used a one-versus-all logistic regression as the occupation classifier; to construct the DNN representation, we started with word embeddings as input and then trained a DNN to predict occupations in an end-to-end fashion. For each representation, we considered two scenarios: (1) where explicit gender indicators---e.g., first names and pronouns---are available to the classifier, (2) where explicit gender indicators are ``scrubbed.'' For example, these scenarios correspond to predicting the occupation \textit{nurse} from the text \textit{[She] graduated from Lehigh University, with honours in 1998. [Nancy] has years of experience in weight loss surgery, patient support, education, and diabetes,} with and without the bracketed words.\looseness=-1

\subsection{Semantic Representations} 
\label{subsec:rep}

\paragraph{Bag-of-words} The BOW representation encodes the $i^{\textrm{th}}$ biography as a sparse vector $x_i^{\textrm{BOW}}$. Each element of this vector corresponds to a word type in the vocabulary, equal to 1 if the biography contains a token of this type and 0 otherwise. Despite recent successes of using more complex semantic representations for document classification, the BOW representation provides a good baseline and is still widely used, especially in scenarios where interpretability is important. To predict occupations, we trained a one-versus-all logistic regression with $L_2$ regularization using our dataset's training split represented using the BOW representation. 

\paragraph{Word embeddings} The WE representation encodes the $i^{\textrm{th}}$ biography as a vector $x_i^{\textrm{WE}}$, obtained by averaging the \texttt{fastText} word embeddings~\cite{bojanowski2017enriching,mikolov2018advances} for the word types present in that biography.\footnote{We note that the \texttt{fastText} word embeddings were trained using the Common Crawl, albeit using a different subset than the one we used to create our dataset.} The WE representation is surprisingly effective at capturing non-trivial semantic information~\cite{adi2016fine}. To predict occupations, we trained a one-versus-all logistic regression with $L_2$ regularization using our dataset's training split represented using the WE representation.\looseness=-1

\paragraph{Deep recurrent neural networks} To construct the DNN representation, we started with the \texttt{fastText} word embeddings as input and then
trained a DNN to predict occupations in an end-to-end fashion. We used an architecture similar to that of~\citet{yang2016hierarchical}, but with just one bi-directional recurrent neural network at the level of words and with gated recurrent units (GRUs)~\cite{chung2014empirical} instead of long short-term memory units; this model uses an attention mechanism---an integral part of modern neural network architectures~\cite{vaswani2017attention}. Our choice of architecture was motivated by a desire to use a relatively simple model that would be easy to interpret.\looseness=-1

Formally, given the $i^{\textrm{th}}$ biography represented as a sequence of tokens $w_i^1, \ldots, w_i^T$, we start by replacing each token $w_i^t$ with the \verb|fastText| word embedding for that word type to yield $e_i^1, \ldots, e_i^T$. The DNN then uses a GRU to process the biography in both forward and reverse directions and concatenates the corresponding hidden states from both directions to re-represent the $t^{\textrm{th}}$ token as follows:
\begin{align}
\overrightarrow{h_i^t} &= \overrightarrow{GRU}(e_i^t, h_i^{t-1}) \\
\overleftarrow{h_i^t} &= \overleftarrow{GRU}(e_i^t,h_i^{t+1}) \\
h_i^t &= [\overleftarrow{h_i^t};\overrightarrow{h_i^t}].
\end{align}
Next, the DNN projects each hidden state $h_i^t$ to the attention dimension $k_a$ via a fully connected layer with weights $W_a$ and $b_a$, and transforms the result into an unnormalized scalar $u_i^t$ via a vector $w_a$:\looseness=-1
\begin{align}
\hat{u}_i^t &= \tanh\,(W_a\, h_i^t + b_a) \\
u_i^t &= w_a^\intercal \hat{u}_i^t.
\end{align}
Each scalar is then normalized to yield an attention weight:
\begin{equation}
\alpha_i^t = \frac{\exp\,(u_i^t)}{\sum_{t'=1}^T \exp\,(u_i^{t'})}
\label{eq:attetnion_alpha}.
\end{equation}
Finally, we obtain the DNN representation via a weighted sum:
\begin{equation}
x_i^{\textrm{DNN}} = \sum_{t=1}^T \alpha_i^t\,h_i^t.
\end{equation}
The DNN makes predictions as follows:
\begin{equation}
\hat{y}_i = \textrm{softmax}(W_0\, x_i^{\textrm{DNN}} + b_0),
\end{equation}
where $\hat{y}_i$ is the predicted occupation for the $i^{\textrm{th}}$ biography.

We trained the DNN using our dataset's training split and a standard cross-entropy loss applied to the output of the last layer.

\subsection{Explicit Gender Indicators}
\label{sec:explicit}

For each semantic representation, we considered two scenarios. In the first scenario, the representation included all word types, meaning that explicit gender indicators are available to the occupation classifier. In the second scenario, we ``scrubbed'' explicit gender indicators prior to creating the representation, meaning that these indicators are not available to the occupation classifier. Specifically, we deleted the subject's first name, along with the words \textit{he}, \textit{she}, \textit{her}, \textit{his}, \textit{him}, \textit{hers}, \textit{himself}, \textit{herself}, \textit{mr}, \textit{mrs}, and \textit{ms} from each biography.

\section{Analysis and Results}
\label{sec:results}

In this section, we analyze the potential allocation harms that can result from semantic representation bias. To do this, we study the performance of the occupation classifier for each semantic representation, with and without explicit gender indicators, as described in the previous section. The classifiers' overall accuracies are shown in Figure~\ref{fig:acc}. We start by analyzing gender bias for the scenario in which the semantic representations include all word types, including explicit gender indicators. We then analyze gender bias in the scenario in which explicit gender indicators are ``scrubbed,'' and use the DNN's per-token attention weights to understand proxy behavior that occurs in the absence of explicit gender indicators.\looseness=-1

\begin{figure}[t]
\includegraphics[width=\linewidth]{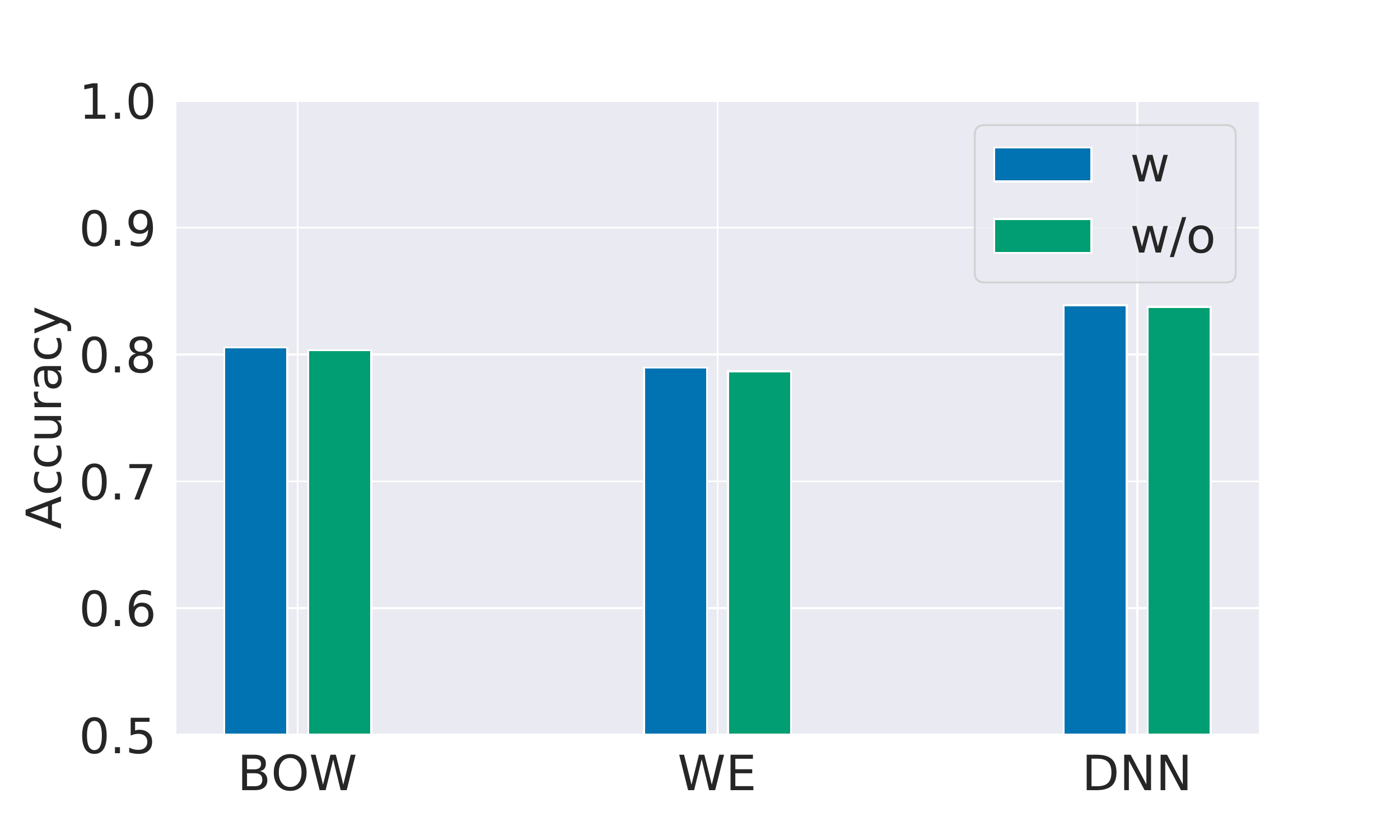}
\caption{Occupation classifier accuracy for each semantic representation, with and without explicit gender indicators.}
\label{fig:acc}
\end{figure}

\subsection{With Explicit Gender Indicators}

\paragraph{True positive rate gender gap} For each semantic representation, we quantify gender bias by using our dataset's testing split to calculate the occupation classifier's TPR gender gap---i.e., the difference in TPRs between binary genders $g$ and ${\sim} g$---for each occupation $y$:
\begin{align}
\textrm{TPR}_{g,y} &= P\,[\hat{Y} = y \,|\, G = g, Y=y]\\
\textrm{Gap}_{g,y} & = \textrm{TPR}_{g,y} - \textrm{TPR}_{{\sim}g,y},
\label{eg:gendergap}
\end{align}
where $\hat{Y}$ and $Y$ are random variables representing the predicted and target labels (i.e., occupations) for a biography and $G$ is a random variable representing the binary gender of the biography's subject.

Defining the percentage of people with gender $g$ in occupation $y$ as $\pi_{g,y} = P\,[G=g \,|\, Y=y]$, Figure~\ref{image:bow_gap} shows $\textrm{Gap}_{\textrm{female},y}$ versus $\pi_{\textrm{female},y}$ for each occupation $y$ for the BOW representation with explicit gender indicators; Figure~\ref{image:gender_gap} depicts the same information for all three representations, with and without explicit gender indicators.\looseness=-1

\begin{figure}[ht]
\includegraphics[width=\linewidth]{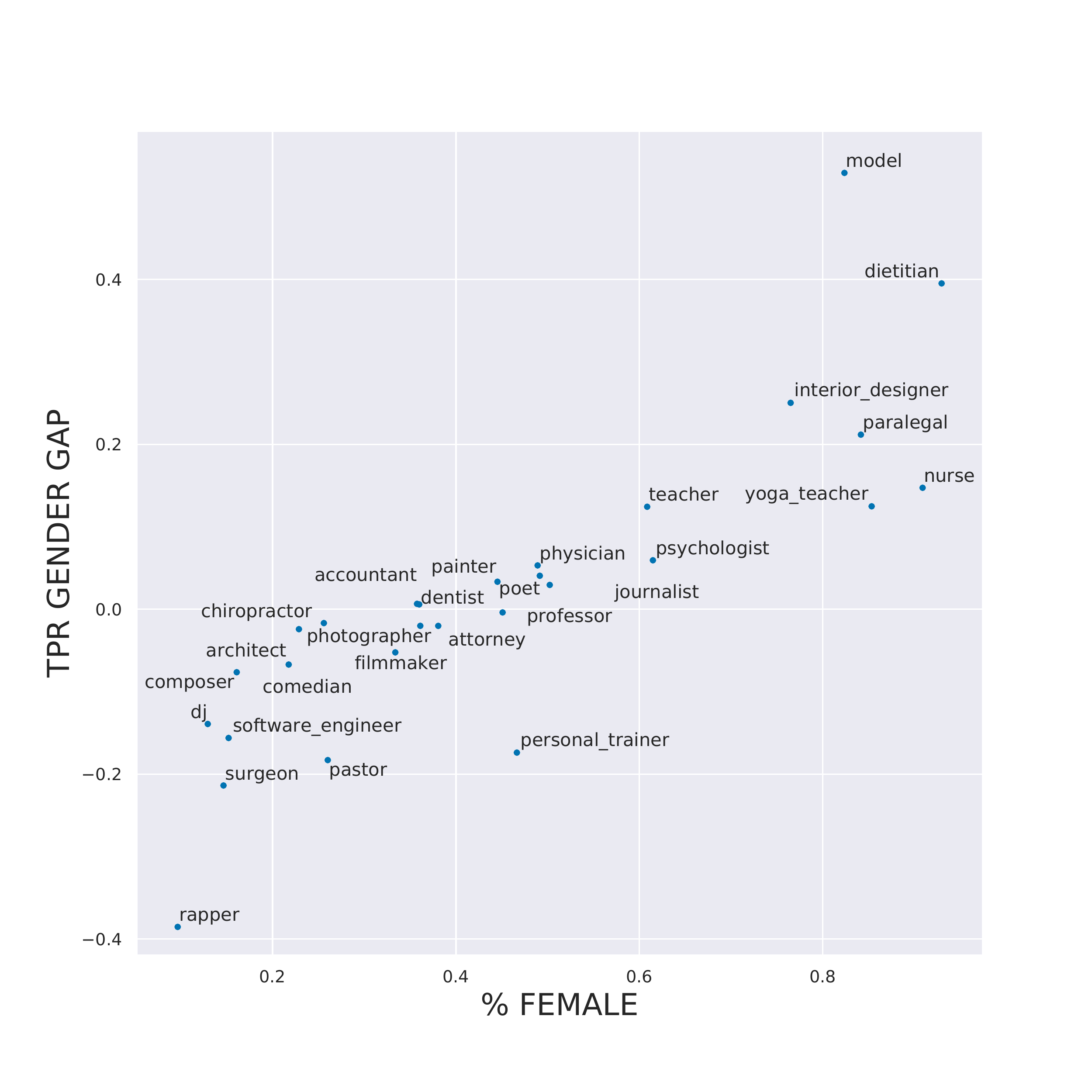}
\caption{$\textrm{Gap}_{\textrm{female},y}$ versus $\pi_{\textrm{female},y}$ for each occupation $y$ for the BOW representation with explicit gender indicators.}
\label{image:bow_gap}
\end{figure}

\begin{figure*}[t]
\includegraphics[width=\linewidth]{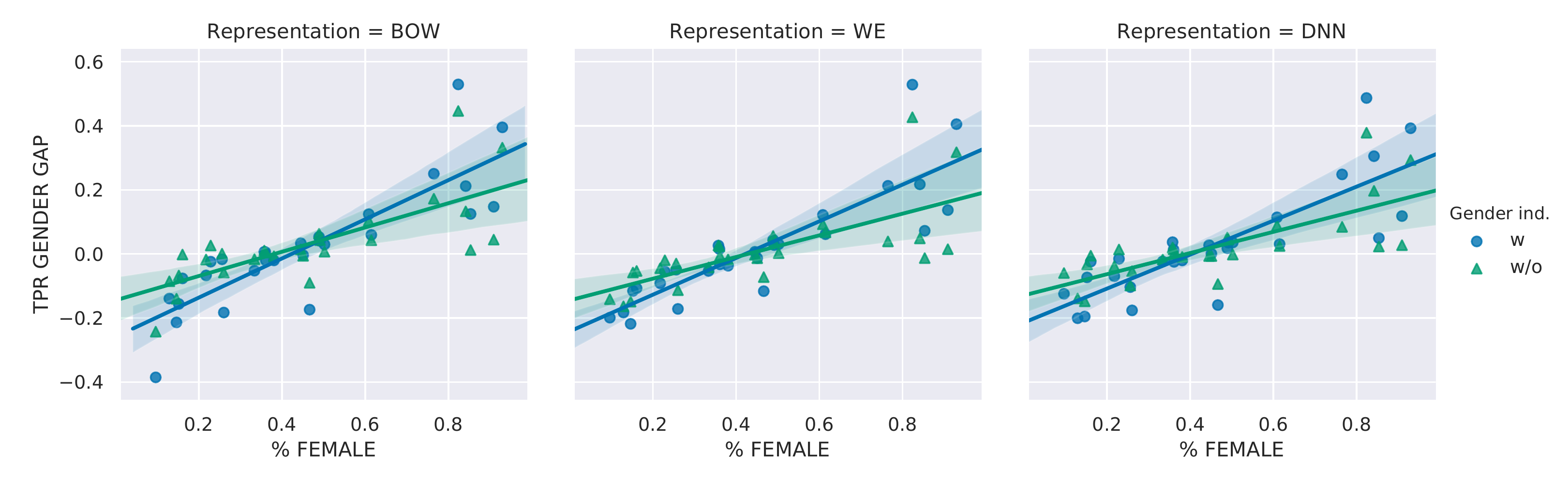}
\caption{$\textrm{Gap}_{\textrm{female},y}$ versus $\pi_{\textrm{female},y}$ for each occupation $y$ for all three semantic representations, with and without explicit gender indicators. Correlation coefficients: BOW-w 0.85; BOW-wo 0.74; WE-w 0.86; WE-wo 0.71; DNN-w 0.82, DNN-wo 0.74.\looseness=-1}
\label{image:gender_gap}
\end{figure*}

\paragraph{Compounding imbalance} We define the gender imbalance of occupation $y$ as $\frac{\pi_{g,y}}{\pi_{\sim g, y}}$; gender $g$ is underrepresented if $\frac{\pi_{g,y}}{\pi_{\sim g, y}} < 1$ or, equivalently, if $\pi_{g,y}<0.5$. The gender imbalance is compounded if the underrepresented gender has a lower TPR than the overrepresented gender---e.g., if $\textrm{Gap}_{g,y} < 0$ and $g$ is underrepresented.\looseness=-1

\begin{thm}
If  $\pi_{g,y}<0.5$ and $\textrm{Gap}_{g,y}<0$, then
\begin{equation}
P\,[G=g\,|\,Y=\hat{Y}=y]<\pi_{g,y}.
\end{equation}
\end{thm}
\begin{proof}
Via Bayes theorem,
\begin{equation}
P\,[G=g\,|\,Y=\hat{Y}=y]= \frac{\pi_{g,y} \,\textrm{TPR}_{g,y}}{P\,[\hat{Y}=y\,|\,Y=y]}.
\end{equation}
If $\pi_{g,y}<\pi_{{\sim}g,y}$ and $\textrm{TPR}_{g,y}<\textrm{TPR}_{{\sim}g,y}$, then
\begin{equation}
 \frac{P\,[G=g\,|\,Y=\hat{Y}=y]}{P\,[G={\sim}g\,|\,Y=\hat{Y}=y]} = \frac{\pi_{g,y}\,\textrm{TPR}_{g,y}}{\pi_{{\sim}g,y}\,\textrm{TPR}_{{\sim}g,y}}< \frac{\pi_{g,y}}{\pi_{{\sim}g,y}},
\end{equation}
so the gender imbalance for the true positives in occupation $y$ is larger than the initial gender imbalance in that occupation.
\end{proof}

As explained in Section~\ref{sec:background}, if the initial gender imbalance is due to prior injustices, an occupation classifier will compound these injustices, which may correspond to indirect discrimination~\citep{hellman2017indirect}. 

It is clear from Figure~\ref{image:bow_gap} that there are few occupations with an equal percentage of men and women---i.e., almost all occupations have a gender imbalance---and that for that for occupations in which women (conversely men) are underrepresented, $\textrm{Gap}_{\textrm{female},y} < 0$ (conversely $\textrm{Gap}_{\textrm{male},y} < 0$). In other words, there is a positive correlation between the TPR gender gap for an occupation $y$ and the gender imbalance in that occupation. (Figure~\ref{image:gender_gap} illustrates that this is also the case for the WE and DNN representations.) As a result, if the occupation classifier for the BOW representation were used to recruit candidates for jobs in occupation $y$, it would compound the gender imbalance by a factor of $\frac{\textrm{TPR}_{g,y}}{\textrm{TPR}_{{\sim} g, y}}$, where $g$ is the underrepresented gender. For example, $14.6\%$ of the surgeons in our dataset's testing split are women---i.e., $\pi_{\textrm{female},\textrm{surgeon}} < 0.5$. The classifier for the BOW representation is able to correctly predict that $71\%$ of male surgeons and $54.5\%$ of female surgeons are indeed surgeons---i.e., $\textrm{Gap}_{\textrm{female},\textrm{surgeon}} < 0$. Consequently, only $11.6\%$ of the true positives are women, so the gender imbalance is compounded.



\paragraph{Counterfactuals} To isolate the effects of explicit gender indicators on the representations' occupation classifiers, we examined differences between the classifiers' predictions on our dataset's testing split as described above and their predictions on our dataset's testing split with first names removed and other explicit gender indicators (see Section~\ref{sec:explicit}) swapped for their complements, keeping everything else the same. This analysis is similar in spirit to the experiment of \citet{bertrand2004emily}, in which they responded to help-wanted ads using fictitious resumes in order to measure the effect of gender and race on the likelihood of being called for an interview. By analyzing the counterfactuals obtained by swapping gender indicators, we can answer the question, ``Which occupation would this classifier predict if this biography had used indicators corresponding to the other gender.'' This question is interesting because we would expect an occupation classifier to predict the same occupation for a man and a woman with identical biographies. We note that this question is not the same as the question, ``Which occupation would this classifier predict if this biography's subject were the other gender.'' Although the latter question is arguably more interesting, it cannot be answered without additionally changing all other factors that are correlated with gender~\citep{kilbertus2017avoiding}. 


For the BOW representation, we find that the classifier's predictions for $5.5\%$ of the biographies in our testing split change when their gender indicators are swapped; for the WE and DNN representations, these percentages are $12.2\%$ and $4.6\%$, respectively. To better understand the effects of explicit gender indicators on the classifiers' predictions, we consider pairs of occupations. Specifically, for each gender $g$ and pair of occupations $(y^1, y^2)$, we identify the set of biographies that are incorrectly predicted as having occupation $y^1$ with their original gender indicators, but correctly predicted as having occupation $y^2$ when their gender indicators are swapped:
\begin{equation}
\mathbb{S}_{g,(y^1,y^2)} = \{x^R_i: \hat{y}_i = y^1,\hat{y}^{(g\leftrightarrow {\sim} g)}_i = y^2, y_i = y^2 \},
\end{equation}
where $x^R_i$ is the $i^{\textrm{th}}$ biography, $y_i$ is the target label (i.e., occupation) for that biography, $\hat{y}_i$ is the predicted label for that biography with its original gender indicators, and $\hat{y}^{(g \leftrightarrow {\sim} g)}_i$ is the predicted label for that biography when its gender indicators are swapped. For example, $\mathbb{S}_{\textrm{female},(\textrm{nurse},\textrm{surgeon})}$ is the set of biographies for female surgeons who are incorrectly predicted as nurses, but correctly predicted as surgeons when their biographies use male indicators. We also identify the total set of biographies $\mathbb{S}_{g,y^2}$ that are only correctly predicted as having occupation $y^2$ when their gender indicators are swapped, and then calculate the percentage of these biographies for which the predicted label changes from $y^1$ to $y^2$:
\begin{eqnarray}\label{eq:pi_swap}
\Pi_{g,(y^1,y^2)} = \frac{|\mathbb{S}_{g,(y^1,y^2)} |}{|\mathbb{S}_{g,y^2}|} \times 100\%.
\end{eqnarray}

Tables \ref{table:swap_m} and \ref{table:swap_f} list, for the BOW representation, the five pairs of occupations with the largest values of $\Pi_{g,(y^1,y^2)}$. For example, $7.1\%$ of male paralegals whose occupations are only correctly predicted when their gender indicators are swapped are incorrectly predicted as attorneys when their biographies use male indicators. Similarly, $14.7\%$ of female rappers whose occupations are only correctly predicted when their gender indicators are swapped are incorrectly predicted as models when their biographies use female indicators.

\begin{table}[ht!]
\caption{Pairs of occupations with the largest values of $\Pi_{\textrm{male},(y^1,y^2)}$---i.e., the percentage of men's biographies that are only correctly predicted as $y^2$ when their indicators are swapped for which the predicted label changes from $y^1$.}
\label{table:swap_m}
\begin{tabular}{lll}
\toprule
$y^1$     & $y^2$       & $\Pi_{\textrm{male},(y^1,y^2)}$ \\ \midrule
attorney & paralegal  & $7.1\%$           \\ 
architect  & interior designer   & $4.7\%$             \\ 
professor  & dietitian  & $4.3\%$            \\ 
photographer   & interior designer & $3.5\%$            \\
teacher  & yoga teacher & $3.3\%$              \\ \bottomrule
\end{tabular}
\vspace{0.5cm}
\caption{Pairs of occupations with the largest values of $\Pi_{\textrm{female},(y^1,y^2)}$---i.e., the percentage of women's biographies that are only correctly predicted as $y^2$ when their indicators are swapped for which the predicted label changes from $y^1$.}
\label{table:swap_f}
\begin{tabular}{lll}
\toprule
$y^1$     & $y^2$       & $\Pi_{\textrm{female},(y^1,y^2)}  $ \\ \midrule
model   & rapper   & $14.7\%$             \\
teacher &  pastor   & $8.5\%$           \\ 
professor   &  software engineer  & $6.5\%$            \\
professor   & surgeon & $4.8\%$            \\
physician   & surgeon & $3.8\%$              \\ \bottomrule
\end{tabular}
\end{table}

\subsection{Without Explicit Gender Indicators}
\label{subsec:scrub}

\paragraph{Remaining gender information} If there are no differences between the ways that men and women in occupation $y$ represent themselves in their biographies other than explicit gender indicators, then ``scrubbing'' these indicators should be sufficient to remove all information about gender from the biographies---i.e.,\looseness=-1
\begin{equation}
P\,[\tilde{X}^R = \tilde{x}^R \,|\, G = g, Y = y] = P\,[\tilde{X}^R = \tilde{x}^R \,|\, G = {\sim}g, Y = y],
\end{equation}
where $\tilde{X}^R$ is a random variable representing a biography without explicit gender indicators, $G$ is a random variable representing the binary gender of the biography's subject, and $Y$ is a random variable representing the biography's target label (i.e., occupation). In turn, this would mean that the TPRs for genders $g$ and ${\sim}g$ are identical:
\begin{align}
\textrm{TPR}_{g,y} &= P\,[\hat{Y} = y \,|\, G = g, Y = y]\\
&= P\,[\hat{Y} = y \,|\, G = {\sim} g, Y = y]\\
&= \textrm{TPR}_{{\sim}g,y},
\end{align}
where $\hat{Y} = f(\tilde{X}^R)$ is a random variable representing the predicted label (i.e., occupation) for $\tilde{X}^R$. Moreover, it would also mean that
\begin{equation}
P\,[G = g \,|\, \tilde{X}^R = \tilde{x}^R,Y = y] = P[G = {\sim}g \,|\, \tilde{X}^R = \tilde{x}^R,Y = y],
\end{equation}
making it impossible to predict the gender of a ``scrubbed'' biography's subject belonging to occupation $y$ better than random.

In order to determine whether ``scrubbing'' explicit gender indicators is sufficient to remove all information about gender, we used a balanced subsample of our dataset to predict people's gender. We created a  subsampled training split by first discarding from our dataset's training split all occupations for which there were not at least $1,000$ biographies for each gender. For each of the remaining twenty-one occupations, we then subsampled $1,000$ biographies for each gender to yield $42,000$ biographies, balanced by occupation and gender. To create a subsampled validation split, we first identified the occupation and gender from those represented in the subsampled training split with the smallest number of biographies in our dataset's validation split. Then, we subsampled that number of biographies from our dataset's validation split for each of the twenty-one occupations represented in the subsampled training split and each gender. We created a subsampled testing split similarly. When using the BOW and WE representations, we used a logistic regression with $L_2$ regularization as the gender classifier; to construct the DNN representation, we started with word embeddings as input and then trained a DNN to predict gender in an end-to-end fashion, similar to the methodology described in Section~\ref{sec:methods}.\looseness=-1

Using the subsampled testing split, we find that the gender classifier for the BOW representation has an accuracy of $65.5\%$, while the DNN representation has an accuracy of $68.2\%$. These accuracies are higher than $50\%$, so ``scrubbing'' explicit gender indicators is not sufficient to remove all information about gender. This finding is reinforced by the scatterplot in Figure~\ref{fig:scatter}, which shows log frequency versus correlation with $G = \textrm{female}$ for each word type in the vocabulary. It is clear from this scatterplot that deleting all words that are correlated with gender would not be feasible.\looseness=-1

\begin{figure}[ht]
\includegraphics[width=\linewidth]{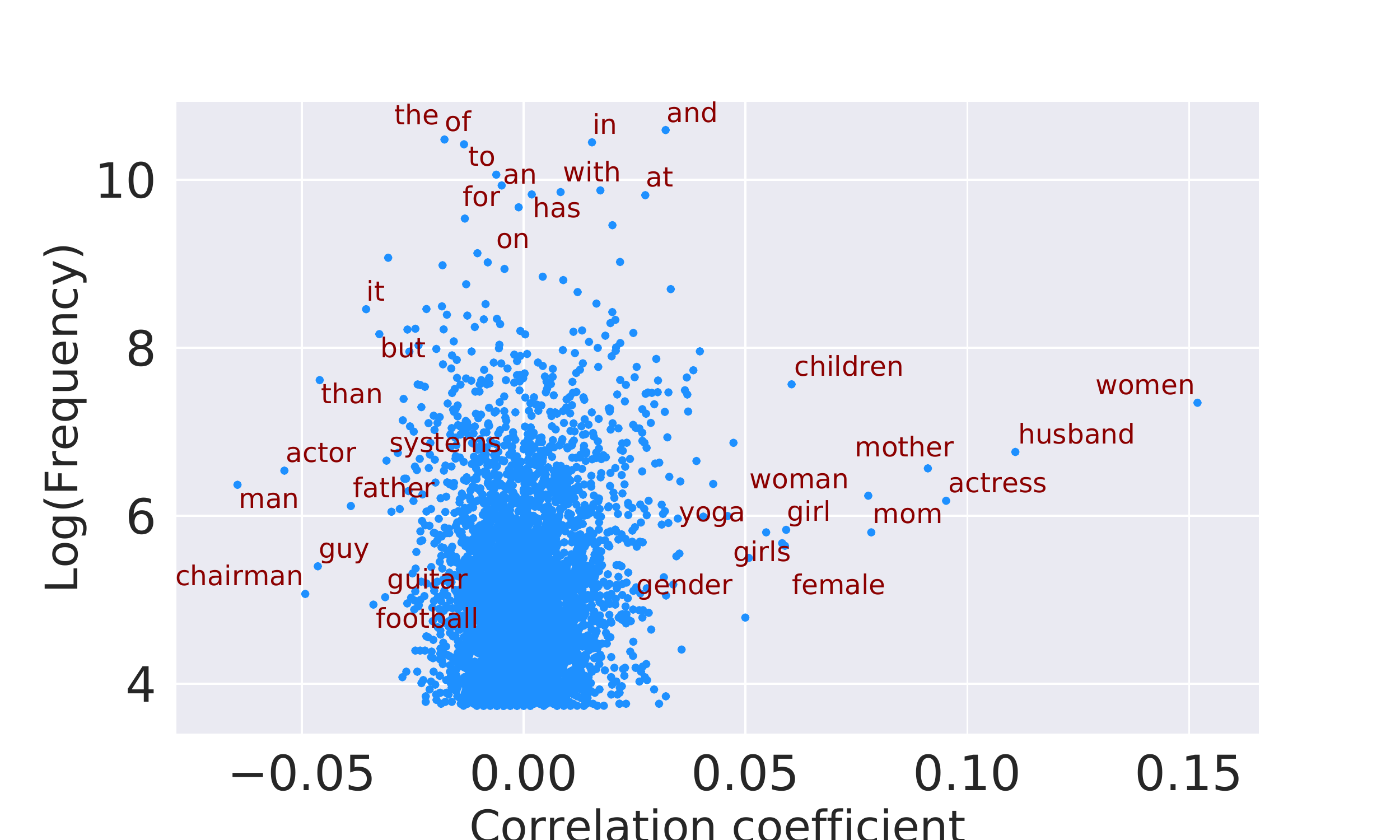}
\caption{Scatterplot of log frequency versus correlation with $G = \textrm{female}$ for each word type in the vocabulary.}
\label{fig:scatter}
\end{figure}

\paragraph{True positive rate gender gap and compounding imbalance} For each semantic representation, we again quantify gender bias by using our (original) dataset's testing split to calculate the occupation classifier's TPR gender gap for each occupation. Figure~\ref{image:gender_gap} shows $\textrm{Gap}_{\textrm{female},y}$ versus $\pi_{\textrm{female},y}$ for each occupation $y$ for all three representations, with and without explicit gender indicators. ``Scrubbing'' explicit gender indicators reduces the TPR gender gaps, while the classifiers' accuracies (shown in Figure~\ref{fig:acc}) remain roughly the same; however, for some occupations, $\textrm{Gap}_{\textrm{female},y}$ is still very large. Moreover, because there is still a positive correlation between the TPR gender gap for an occupation $y$ and the gender imbalance in that occupation, ``scrubbing'' explicit gender indicators will not prevent the classifiers from compounding gender imbalances.\looseness=-1

We note that compounding imbalances are especially problematic if people repeatedly encounter such classifiers---i.e., if an occupation classifier's predictions determine the data used by subsequent occupation classifiers. Who is offered a job today will affect the gender (im)balance in that occupation in the future. If a classifier compounds existing gender imbalances, then the underrepresented gender will, over time, become even further underrepresented---a phenomenon sometimes referred to as the ``leaky pipeline.'' 

To illustrate this phenomenon, we performed simulations using the DNN representation in which the candidate pool at time $t+1$ is defined by the true positives at time $t$. Defining the percentage of people with gender $g$ in occupation $y$ at time $t$ as $\pi_{g,y}^{(t)}$, we fit a linear regression to the TPR gender gaps for different values of $\pi_{g,y}^{(t)}$:\looseness=-1
\begin{equation}\label{eq:reg_gap}
\widehat{\textrm{Gap}}_{g,y}^{(t)} = \pi_{g,y}^{(t)}\,\beta_1 + \beta_0.
\end{equation}
Using this regression model, we are then able to calculate the percentage of people with gender $g$ in occupation $y$ at time $t+1$:
\begin{equation}\label{eq:pi_comp}
\pi_{g,y}^{(t+1)} = \frac{\pi_{g,y}^{(t)}\,\textrm{TPR}_{g,y}^{(t)}}{\pi_{{\sim}g,y}^{(t)}\,(\textrm{TPR}_{g,y}^{(t)}+\textrm{Gap}_{g,y}^{(t)} )+\pi_{g,y}^{(t)}\,\textrm{TPR}_{g,y}^{(t)}}. 
\end{equation}

Figure \ref{fig:comp} shows $\pi_{g,y}^{(t)}$ for $t=0, \ldots, 10$; each subplot corresponds to a different initial gender imbalance. Over time, the gender imbalances compound. We note that there are many different TPR pairs $\textrm{TPR}^{(t)}_{g,y}$ and $\textrm{TPR}^{(t)}_{{\sim}g,y}$ that can result in a given TPR gender gap $\textrm{Gap}^{(t)}_{g,y}$. For example, a TPR gender gap of $-0.2$ might correspond to $0.6-0.8$ or to $0.7-0.9$. Moreover, different TPR pairs will result in different percentages of people with gender $g$ in occupation $y$ at time $t+1$. The bands in Figure~\ref{fig:comp} therefore reflect these differences.

 \begin{figure}[ht]
 \includegraphics[width=\linewidth, trim={0 0 3cm 0},clip]{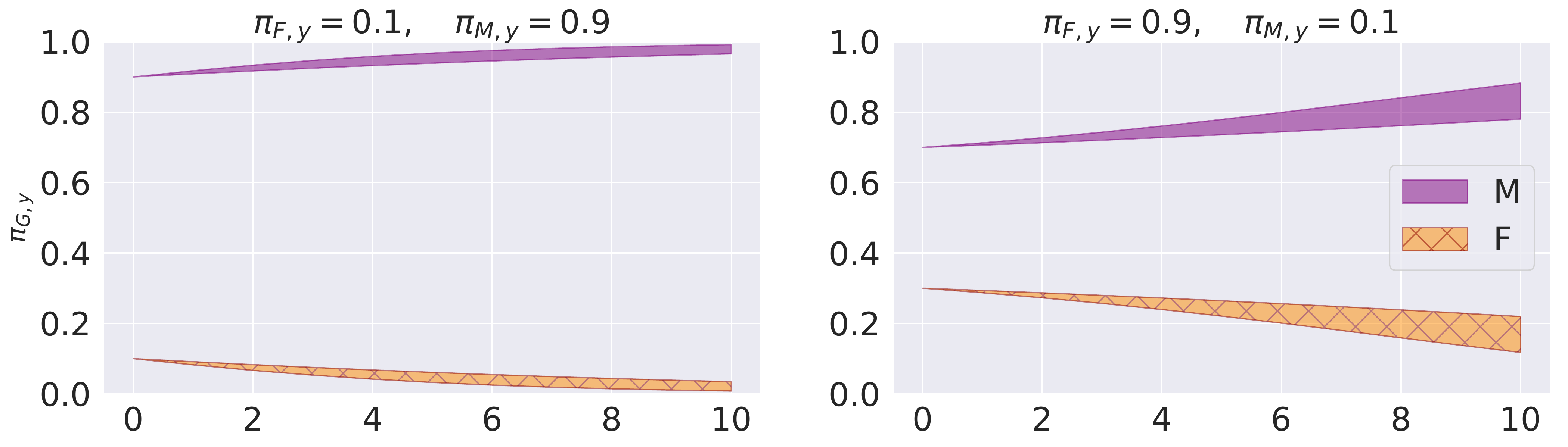}
\caption{Simulations of compounding imbalances using the DNN representation. Each subplot corresponds to a different initial gender imbalance and shows $\pi_{g,y}^{(t)}$ for $t=0, \ldots, 10$.}
\label{fig:comp}
\end{figure}

\paragraph{Attention to gender} The DNN's per-token attention weights allow us to understand proxy behavior that occurs in the absence of explicit gender indicators. The attention weights indicate which tokens are most predictive. For example, Figure~\ref{fig:bill_gates_attention} depicts the 
per-token attention weights from the occupation classifier for the DNN representation when predicting Bill Gates' occupation from an excerpt of his biography on Wikipedia; the larger the weight, the stronger the color. The attention weights for the words \textit{software} and \textit{architect} are very large, and the DNN predicts \textit{software engineer}.

\begin{figure}[t]
\centering
\fbox{
\includegraphics[width=0.9\linewidth, trim={1.2cm 16cm 1.2cm 1.1cm}, clip]{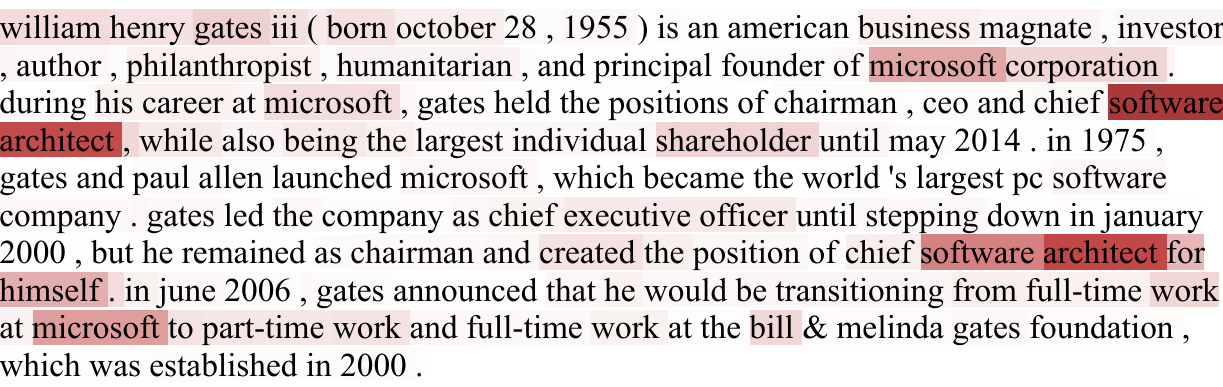}
}
\caption{Visualization of the DNN's per-token attention weights. Predicted label (i.e., occupation): \textit{software engineer}.}
\label{fig:bill_gates_attention}
\end{figure}

In order to understand proxy behavior that occurs in the absence of explicit gender indicators, we first used the subsampled testing split, described above, to obtain per-token attention weights from the gender classifier for the DNN representation. We then used these weights to find ``proxy candidates''---i.e., the words that are most predictive of gender in the absence of explicit gender indicators. Specifically, we computed the sum of the per-token attention weights for each word type, and then selected the types with the largest sums as ``proxy candidates.'' Across multiple runs, we found that the words \textit{women}, \textit{husband}, \textit{mother}, \textit{woman}, and \textit{female} (ordered by decreasing total attention) were consistently ``proxy candidates.''\looseness=-1

\begin{figure*}[ht!]
\centering
\includegraphics[width=0.95\linewidth]{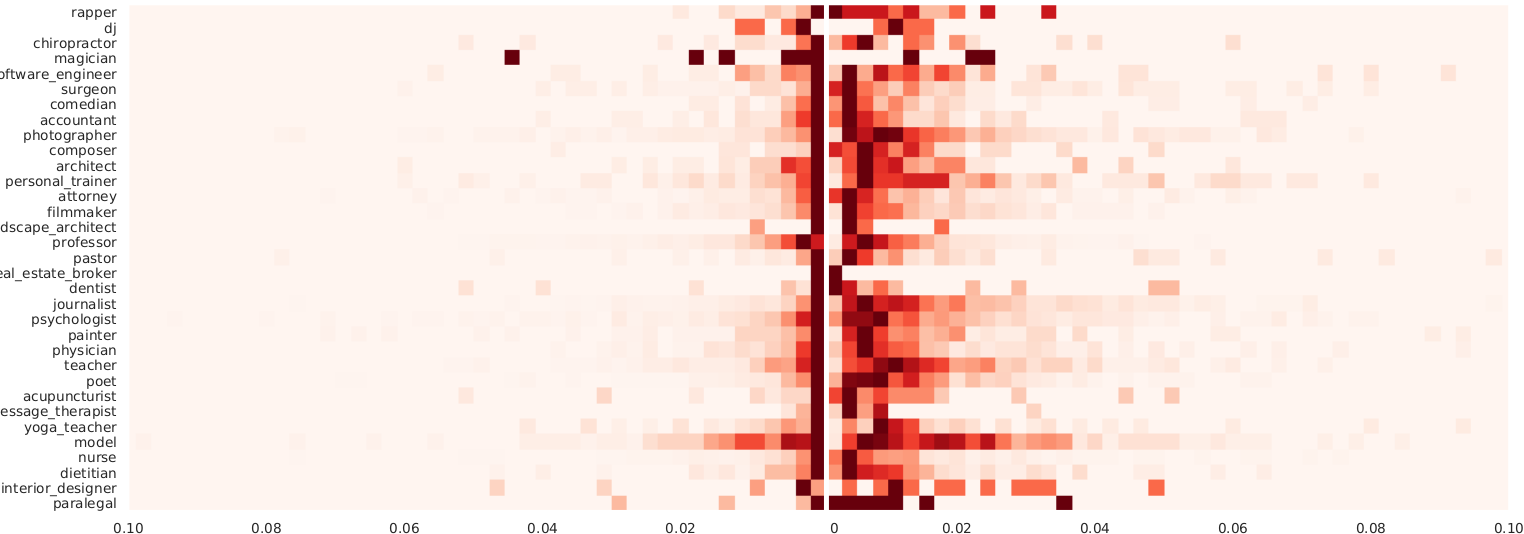}
\caption{Per-occupation histograms of the per-token attention weights from the DNN representation's occupation classifier for the word \textit{women}, with (left) and without (right) explicit gender indicators; occupations are ordered by TPR gender gap.}
\label{fig:attention_women}
\end{figure*}

For each ``proxy candidate,'' we then used our dataset's testing split, with and without explicit gender indicators, to create histograms of the per-token attention weights from the occupation classifier for the DNN representation. These histograms represent the extent to which that ``proxy candidate'' is predictive of occupation, with and without gender indicators. By comparing the histograms for each ``proxy candidate,'' we are able to identify words that are used as proxies for gender in the absence of explicit gender indicators: if there is a big difference between the histograms, then the ``proxy candidate'' is likely a proxy. Figure~\ref{fig:attention_women} shows per-occupation histograms for the word \textit{women}, with (left) and without (right) explicit gender indicators. It is clear that in the absence of explicit gender indicators, the classifier has larger attention weights for the word \textit{women} for all occupations. We see similar behavior for the other ``proxy candidates,'' suggesting that the classifier uses proxies for gender in the absence of explicit gender indicators.

The occupations in Figure~\ref{fig:attention_women} are ordered by TPR gender gap from negative to positive. For occupations in the middle, where there are small or no TPR gender gaps, the classifier still has non-zero attention weights for the word \textit{women}. This means that using gender information does not necessarily lead to a TPR gender gap. We also note that it's possible that the classifier is using gender information to differentiate between occupations with very different gender imbalances that are otherwise similar, such as physician and surgeon.\looseness=-1

\section{Discussion and Future Work}
\label{sec:discussion}

In this paper, we presented a large-scale study of gender bias in occupation classification using a new dataset of hundreds of thousands of online biographies. We showed that there are significant TPR gender gaps when using three different semantic representations: bag-of-words, word embeddings, and deep recurrent neural networks. We also showed that the correlation between these TPR gender gaps and existing gender imbalances in occupations may compound these imbalances. By performing simulations, we demonstrated that compounding imbalances are especially problematic if people repeatedly encounter  occupation classifiers because the underrepresented gender will become even further underrepresented.\looseness=-1

Recently, \citet{dwork2018fairness} showed that fairness does not hold under composition, meaning that if two classifiers are individually fair according to some fairness metric, then the sequential use of these classifiers will not necessarily be fair according the same metric. One interpretation of our finding regarding compounding imbalances is that unfairness holds under composition. Understanding why this is the case, especially given that fairness does not hold under composition, is an interesting direction for future work.

We found that the TPR gender gaps are reduced by ``scrubbing'' explicit gender indicators, while the classifiers' overall accuracies remain roughly the same. This constitutes an empirical example where there is little tradeoff between 
promoting fairness---in this case by ``scrubbing'' explicit gender indicators---and performance. This also constitutes an empirical example where fairness is improved by ``scrubbing'' sensitive attributes, contrary to other examples in the literature~\citep{kleinberg2018algorithmic}. That said, in the absence of explicit gender indicators, we did find that (1) we were able to predict the gender of a biography's subject better than random, even when controlling for occupation; (2) significant TPR gender gaps remain for some occupations; (3) there is still a positive correlation between the TPR gender gap for an occupation and the gender imbalance in that occupation, so existing gender imbalances may be compounded. These findings indicate that there are differences between men's and women's online biographies other than explicit gender indicators. These differences may be due to the ways that men and women represent themselves or due to men and women having different specializations within an occupation. Our findings highlight both the risks of using machine learning in a high-stakes setting and the difficulty of trying to promote fairness by ``scrubbing'' sensitive attributes.\looseness=-1

Our future work will focus primarily on understanding how best to mitigate TPR gender gaps and compounding imbalances in online recruiting and automated hiring. Finally, although we focused on gender bias, we note that
other biases, such as those involving race or socioeconomic status, may also be present in occupation classification. Our methodology and analysis approach may prove useful for quantifying such biases, provided relevant group membership information is available. Moreover, quantifying such biases is an important direction for future work---it is likely that they exist and, in the absence of evidence that they do not, online recruiting and automated hiring run the risk of compounding prior injustices.\looseness=-1

\newpage
\section{Appendix}

\appendix

\section{True positive rate gender gaps across representations}

Figure~\ref{image:bow_gap_wo} shows TPR gender gaps for BOW trained without gender indicators. Figures~\ref{image:we_gap} and ~\ref{image:we_gap_wo} show the results for WE, with and without gender indicators, respectively. Figures~\ref{image:dnn_gap} and ~\ref{image:dnn_gap_wo} show the results for DNN, with and without gender indicators, respectively.

\begin{figure}[ht]
\includegraphics[width=\linewidth]{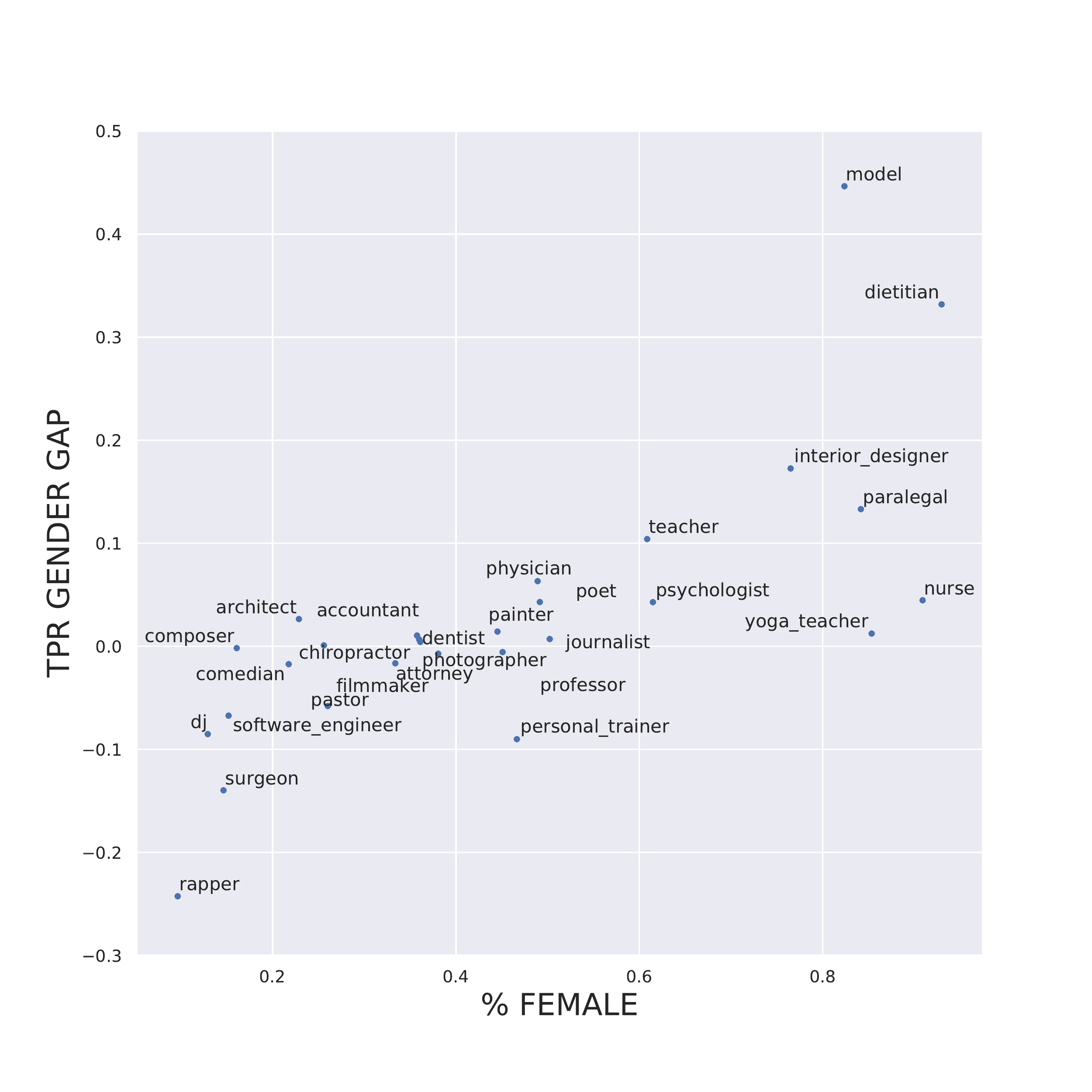}
\caption{Gender gap per occupation vs. $\%$ females in occupation for BOW trained without gender indicators.}
\label{image:bow_gap_wo}
\end{figure}

\begin{figure}[ht]
\includegraphics[width=\linewidth]{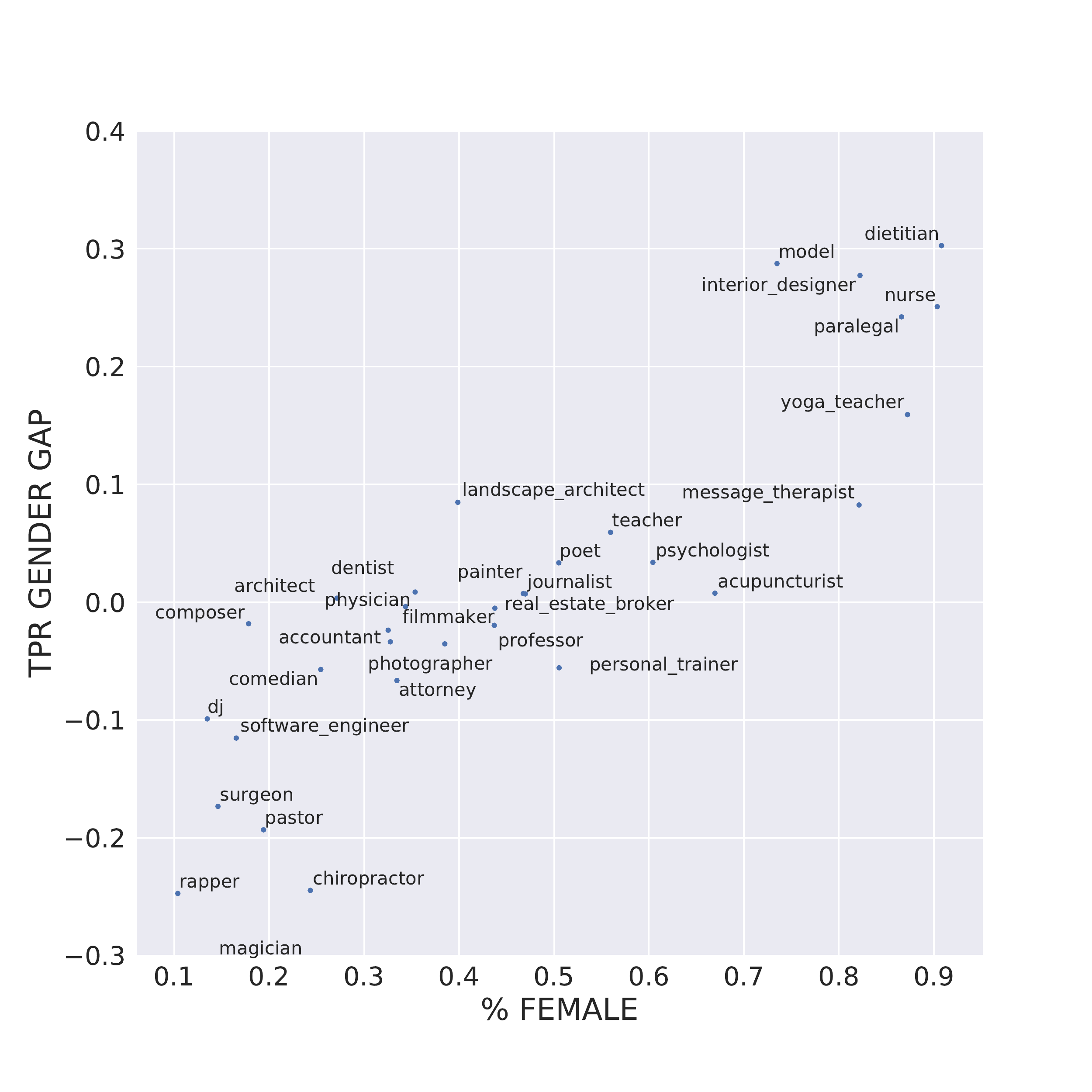}
\caption{Gender gap per occupation vs. $\%$ females in occupation for WE trained with gender indicators.}
\label{image:we_gap}
\end{figure}

\begin{figure}[ht]
\includegraphics[width=\linewidth]{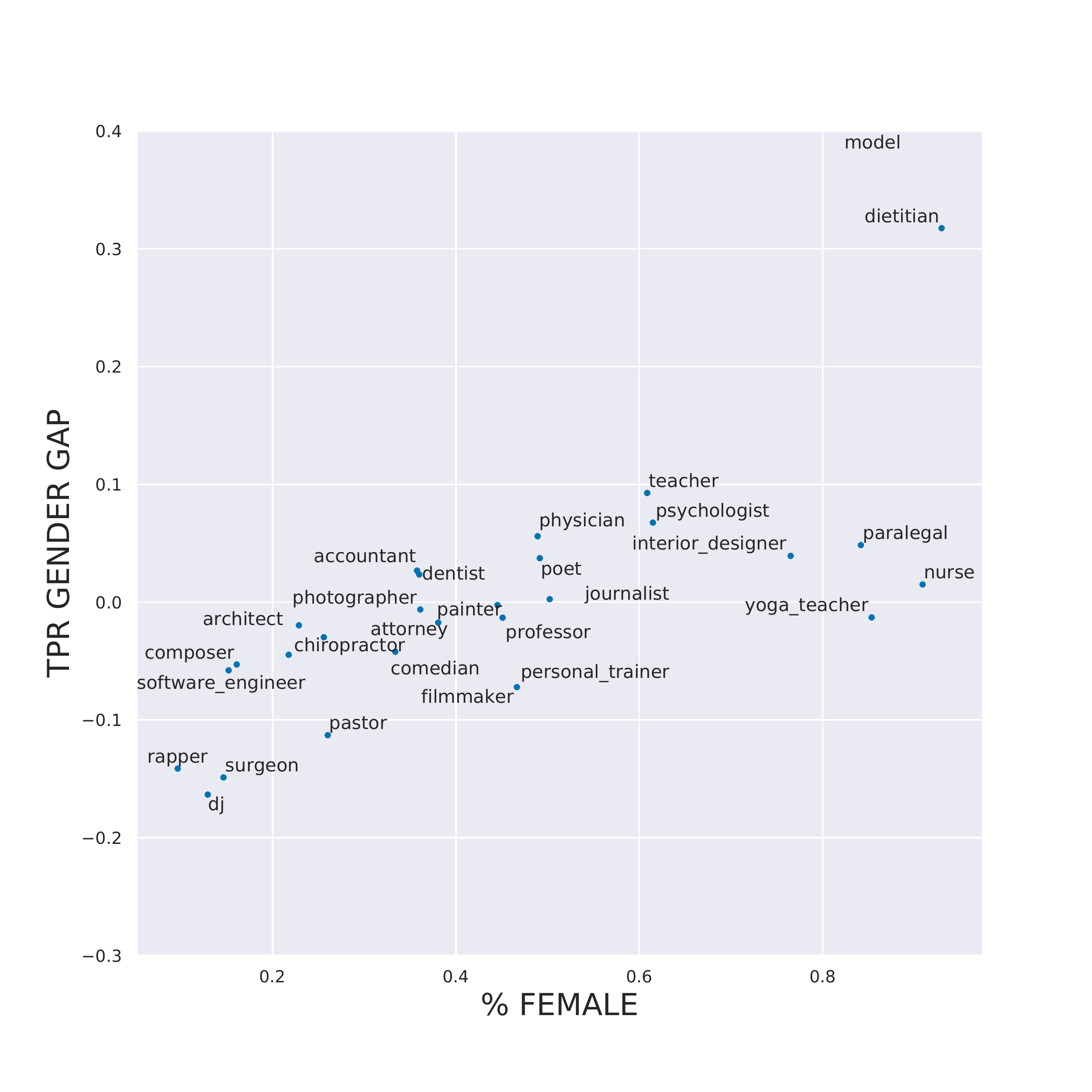}
\caption{Gender gap per occupation vs. $\%$ females in occupation for WE trained without gender indicators.}
\label{image:we_gap_wo}
\end{figure}

\begin{figure}[ht]
\includegraphics[width=\linewidth]{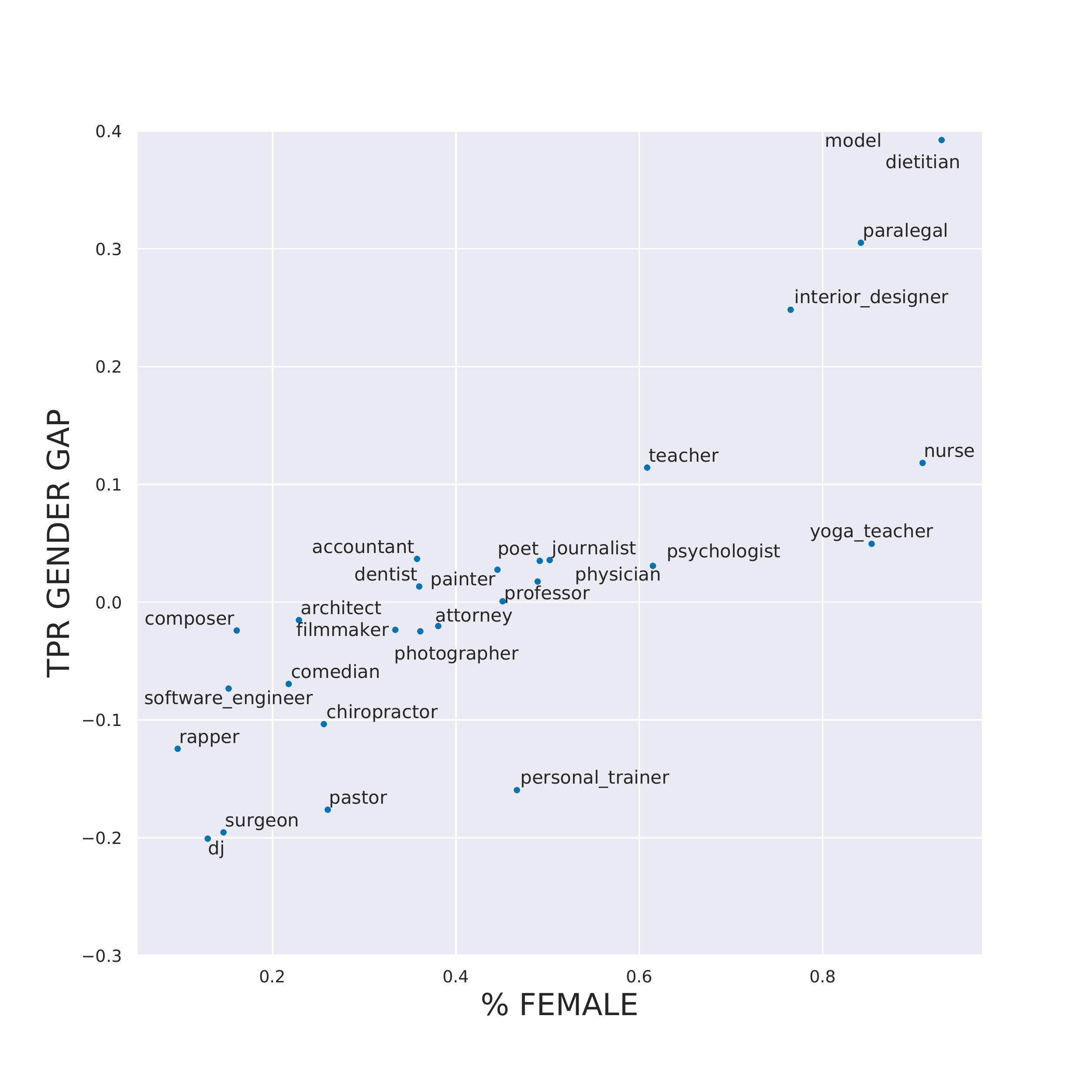}
\caption{Gender gap per occupation vs. $\%$ females in occupation for DNN trained with gender indicators.}
\label{image:dnn_gap}
\end{figure}

\begin{figure}[ht]
\includegraphics[width=\linewidth]{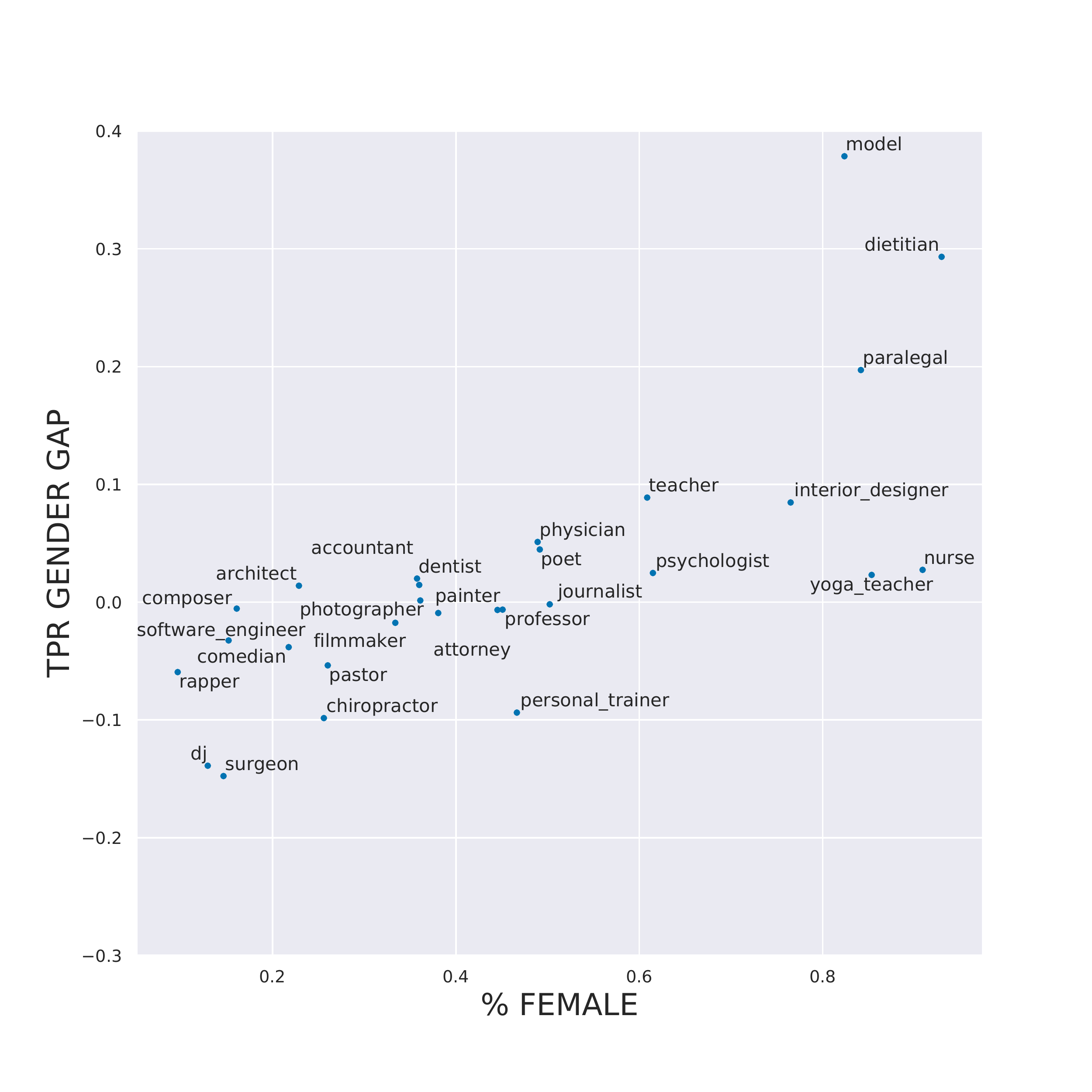}
\caption{Gender gap per occupation vs. $\%$ females in occupation for DNN trained without gender indicators.}
\label{image:dnn_gap_wo}
\end{figure}

\newpage

\section{Attention to gender}
\subsection{Attention to gender proxies}

\autoref{fig:attention_wife_husband} shows the aggregated attention of the DNN model to words ``wife'' and ``husband''. As with the word ``women'', the model trained without gender indicators places more attention on these words. Notice, however, that the shift in attention weights, while it exists, is smaller than for the word ``women'', which is consistent with the lower aggregate attention in the gender prediction model.

\begin{figure*}[h]
\centering
\begin{subfigure}{1\textwidth}
	\centering
	\includegraphics[width=0.9\linewidth]{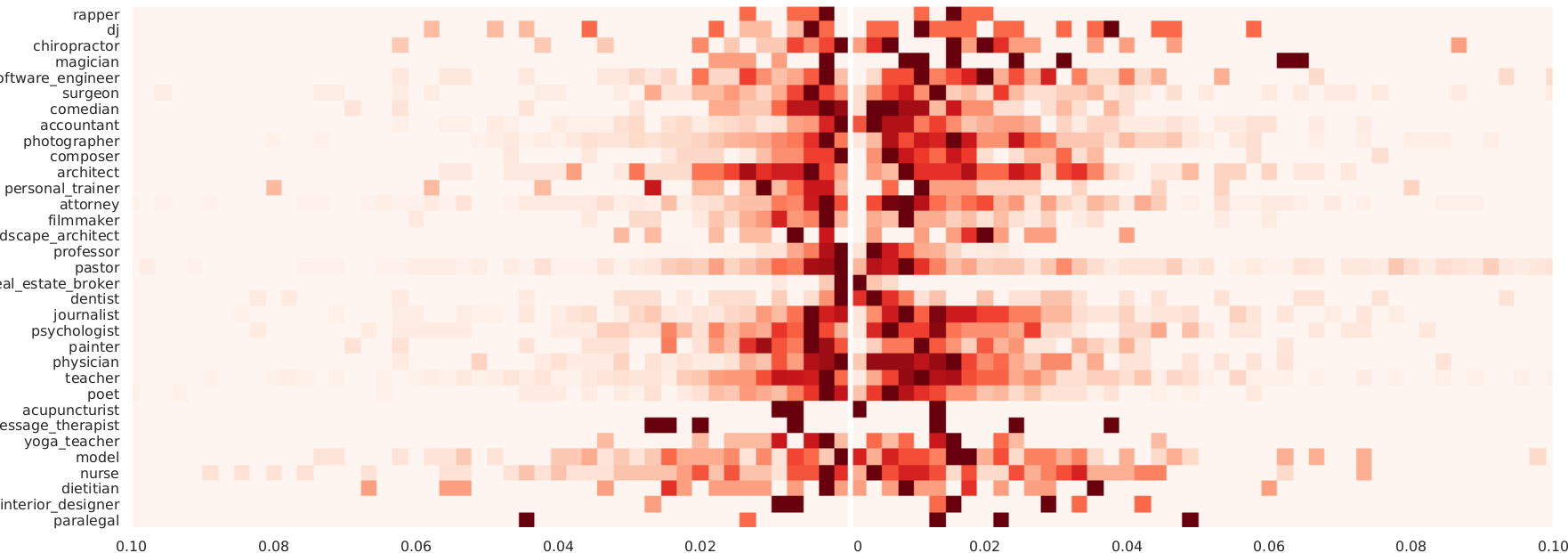}
 	\caption{Aggregated attention to word ``wife''}
 	\label{fig:attention_wife}
\end{subfigure}
\begin{subfigure}{1\textwidth}
	\centering
	\includegraphics[width=0.9\linewidth]{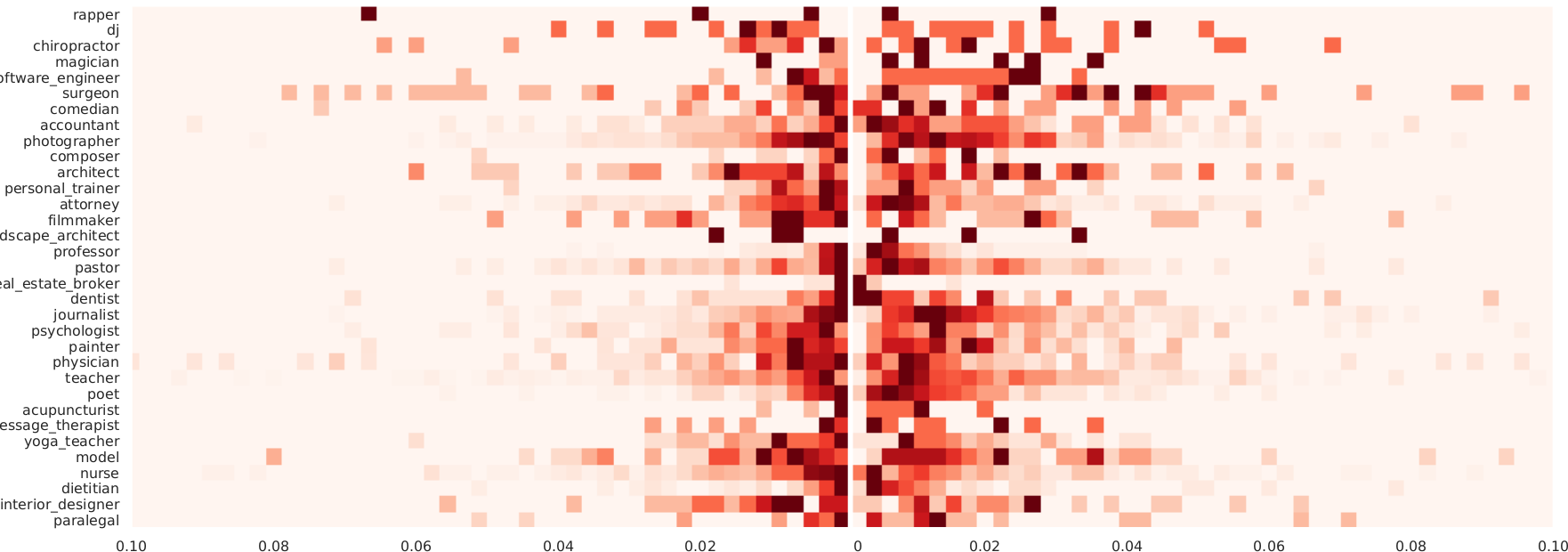}
 	\caption{Aggregated attention to word ``husband''}
 	\label{fig:attention_husband}
\end{subfigure}
\caption{Aggregated attention of DNN to words ``wife'' (\ref{fig:attention_wife}) and ``husband'' (\ref{fig:attention_husband}). In the left, results when model trained with gender indicators. In the right, results when model trained without gender indicators.}
\label{fig:attention_wife_husband}
\end{figure*}

\subsection{Attention to gender indicators}
\autoref{fig:attention_she} shows the attention of the model, trained with and without gender indicators, on the word ``she'' during the prediction of the occupation based on biographies \textit{with} gender indicators. One may expect that in the latter case the model would not attend to this word as it has not seen it during the training. However, the results indicate quite the opposite. In fact, the model puts \textit{much more} attention to it. This can be attributed to the use of word embeddings, which enables the model to learn about words even if it has not explicitly seen them. Interestingly, when exposed to the word ``she'' during prediction, the model seems to receive a stronger gender signal than it has seen during training, and pays a significant amount of attention to it.

\begin{figure*}[h]
\centering
\includegraphics[width=0.9\linewidth]{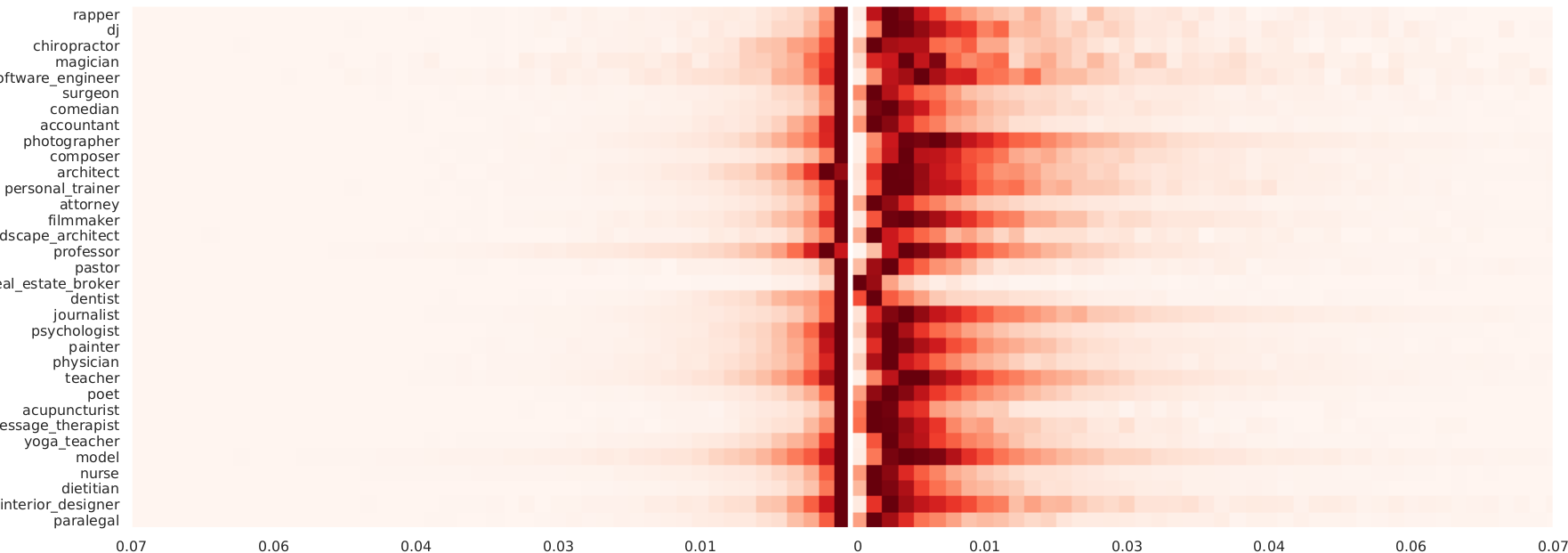}
\caption{Aggregated attention of DNN to word ``she''. In the left, results when model trained with gender indicators. In the right, results when model trained without gender indicators.}
\label{fig:attention_she}
\end{figure*}

\newpage




\bibliographystyle{ACM-Reference-Format}
\bibliography{bib}

\end{document}